\documentclass[11pt]{article}
\usepackage{graphicx}
\usepackage{pgfplots}

\usepackage[colorlinks=true,citecolor=blue]{hyperref}
\usepackage{natbib}
\usepackage{graphicx}
\usepackage{amsfonts}
\usepackage{amsmath}
\usepackage{amssymb}
\usepackage{url}
\usepackage{fancyhdr}
\usepackage{indentfirst}
\usepackage{enumerate}
\usepackage{titlesec}
\usepackage{amsthm}
\usepackage{dsfont}
\usepackage[misc]{ifsym}

\theoremstyle{definition}

\newtheorem{example}{Example}

\newtheorem{axiom}{Axiom}

\newtheorem{property}{Property}

\theoremstyle{plain}
\newtheorem{theorem}{Theorem}

\newtheorem{proposition}{Proposition}

\theoremstyle{remark}
\newtheorem{remark}{Remark}

\usepackage{cases}

\def\lawto{\buildrel \mathrm{d} \over \to}

\theoremstyle{definition}

\def\N{\mathbb{N}}

\def\R{\mathbb{R}}

\def\d{\mathrm{d}}

\usepackage[onehalfspacing]{setspace}

\usepackage{bm}
\usepackage{tikz-qtree}
\usepackage{tikz}

\def\id{\mathds{1}}

\title{Max- and min-stability under first-order stochastic dominance\footnote{This manuscript has not been published, submitted, or otherwise presented elsewhere.}}
\author{Christopher Chambers \and Alan Miller \and Ruodu Wang \and Qinyu Wu}

\author{Christopher P. Chambers\thanks{Department of Economics,
Georgetown University, USA \Letter~{\scriptsize\url{cc1950@georgetown.edu}}}
	\and
	Alan Miller\thanks{
		Faculty of Law, Western University, Canada \Letter~{\scriptsize\url{alan.miller@uwo.ca}}}
	\and
	Ruodu Wang\thanks{Department of Statistics and Actuarial Science, University of Waterloo, Canada \Letter~{\scriptsize\url{wang@uwaterloo.ca}}}
	\and
	Qinyu Wu\thanks{Corresponding author. Department of Statistics and Actuarial Science, University of Waterloo, Canada \Letter~{\scriptsize\url{q35wu@uwaterloo.ca}}}
}

\pgfplotsset{compat=1.18}

\begin{document}

\maketitle

\begin{abstract}  
Max-stability is the property that taking a maximum between two inputs results in a maximum between two outputs.  
We study max-stability with respect to first-order stochastic dominance, the most fundamental notion of stochastic dominance in decision theory. 
Under two additional  standard axioms of nondegeneracy and lower semicontinuity, we establish a representation theorem for functionals satisfying max-stability, which turns out to be represented by the supremum of a bivariate function. A parallel characterization result for min-stability, that is, with  the maximum replaced by the minimum in max-stability, is also established.  By combining both max-stability and min-stability,  we obtain a new characterization for a class of functionals, called the $\Lambda$-quantiles, that appear in finance and political science.



\end{abstract}

\section{Introduction}

First-order stochastic dominance (FSD) stands out as a widely used preorder in decision theory, among many other notions in the literature
 (\cite{QS62}, \cite{HR69,HR71} and \cite{RS70}). 
 For a comprehensive treatment of these notions of stochastic dominance, see \cite{SS07}.

 For two cumulative distribution functions (cdf) $F$ and $G$, denote by $F\vee G$ their maximum, that is, the smallest distribution function that dominates both $F$ and $G$ in FSD. Explicitly, this means $F\vee G(x)=\min\{F(x),G(x)\}$ for $x\in \R$.  
For a functional $\rho$  defined on a set of  distributions, we define max-stability with respect to FSD by 
\begin{align}\label{eq-maxdf}
\rho\left(F\vee G\right)= \rho(F)\vee\rho(G),\mbox{~~~for all $F,G$},
\end{align}
where   $a\vee b$ for two real numbers $a,b$ represents their maximum.
When $\rho$ represents a preference, the property \eqref{eq-maxdf} has the natural interpretation as that either $F$ or $G$ is equally preferable to $F\vee G$.

The max-stability \eqref{eq-maxdf} has been studied by 
\cite{MWW22} in the context of risk management, 
 where it was  called equivalence in model aggregation under FSD. 
 A related notion using point-wise maximum instead of FSD is studied by \cite{KZ21}. 
 In the setting of \cite{MWW22}, $\rho$ is a risk measure, and the aggregate distribution $F\vee G$ is interpreted as a robust (conservative) distributional model in the sense that it dominates both $F$ and $G$ with respect to FSD. The left-hand side of \eqref{eq-maxdf} is the risk evaluation via model aggregation, while the right-hand side is the worst-case risk measure under the set $\{F,G\}$. Therefore, \eqref{eq-maxdf} means the equivalence of two robust risk evaluation approaches, which is motivated by problems in robust optimization. \cite{MWW22} characterized  the benchmark-loss Value-at-Risk (\cite{BBM2020}) through max-stability, lower semicontinuity and translation-invariance.\footnote{A risk measure $\rho$ on a set of distributions satisfies translation-invariance if for $m\in\R$, $G(x)=F(x-m)$ for all $x\in\R$ implies
$\rho(G)=\rho(F)+m$.} 
\cite{MWW22} also considered second-order stochastic dominance (SSD) and characterized the class of the benchmark-adjusted Expected Shortfall (\cite{BMW22}) through max-stability with respect to SSD.\footnote{\color{black} These results are covered in Theorems 4 and 5 of the arXiv version-v1 of \cite{MWW22}.}  
Different from \cite{MWW22},
we focus on FSD and do not assume translation-invariance. This broadens the application of the max-stability with respect to FSD and enhances its advantageous properties for a larger range of functionals. In particular, $\Lambda$-quantile (\cite{FMP14}) satisfies max-stability with respect to FSD and is not translation invariant.

A first main result (Theorem \ref{th-main}) of this paper is a characterization result for \eqref{eq-maxdf}, under two further standard axioms of  nondegeneracy and lower semicontinuity.
The characterized functional can be written as the supremum of a bivariate function satisfying certain further properties. 
This characterization generalizes 
the characterization results in \cite{MWW22}.  \textcolor{black}{In parallel to our work,  \cite{KZ24}  studied functionals satisfying max-stability with respect to general partial orders and their result includes a similar representation to Theorem \ref{th-main}  with different motivations and proof techniques.}

A natural dual   of max-stability, referred to as min-stability, is defined by
\begin{align}\label{eq-mindf}
\rho(F\wedge G)=\rho(F)\wedge \rho(G),\mbox{~~~for all $F,G$,}
\end{align} 
where $F\wedge G$ is the largest distribution function that is dominated by both $F$ and $G$ in FSD; this means $F\wedge G(x)=\max\{F(x),G(x)\}$ for all $x\in\R$. In this context, $F\wedge G$ is interpreted as an optimistic distributional model, and the right-hand side of \eqref{eq-mindf} represents the best-case risk measure under the set $\{F,G\}$. Thus, \eqref{eq-mindf} captures the equivalence of two optimistic risk evaluation approaches. 
Under nondegeneracy and upper semicontinuity, we present a characterization result for min-stability (Proposition \ref{prop-mainMinS}), with a proof similar to that of Theorem \ref{th-main}. 
The joint property of 
max-stability and min-stability has been studied by
\cite{CGR08} in a general setting different from FSD. 
In our setting, combining max-stability and min-stability leads to a characterization for $\Lambda$-quantiles (Theorem \ref{th:main-MM}),  a family of risk measures introduced by \cite{FMP14} in quantitative finance and axiomatized by \cite{BP22}.
The $\Lambda$-quantiles for discrete distributions 
appear in political science as the only voting schemes satisfying  strategy proofness,
 anonymity and 
 efficiency, as shown in the celebrated work of \cite{M80}. 
\textcolor{black}{Our axiomatization   offers a new financial interpretation  of  $\Lambda$-quantiles as risk measures: for any two comonotonic loss random variables, such as losses from two call option price changes on the same asset, if both losses are acceptable then so is their maximum, and if both are unacceptable then so is their minimum.}

The rest of the paper is structured as follows. Section \ref{sec:main} formally introduces the main axiom of max-stability, along with two other standard axioms, and presents the main characterization result. We also present a parallel characterization result for min-stability as a key axiom.
Section \ref{sec:example} provides examples in finance that illustrate the representation outlined in the main result.
Section \ref{sec:lambda} explores the combined framework where both max-stability and min-stability serve as key axioms, which gives a new characterization of $\Lambda$-quantiles.
Section \ref{sec:con} concludes the paper, with a brief discussion on
max-stability or min-stability with respect to other orders  or relations than FSD, which has studied by \cite{K79}
and \cite{KZ21,KZ24}. Section \ref{sec:mainproof} contains all proofs and some technical illustrations.


\section{The main result}\label{sec:main}

\subsection{Axioms}
Let $\mathcal M_C$ be the set of all distributions on $\R$ with compact support, represented by their cdf. We use $\delta_x$ to represent the point-mass at $x$, and $[n]=\{1,\dots,n\}$ for $n\in\N$. Denote by $\id_A$ as the indicator function of $A$.
We adopt the convention that $0\cdot \infty=0$.
Denote by $\mathcal M_{D}$ the set of all discrete distributions with finite support and $\mathcal M_{D,2}$ of all two-point distributions, that is,
\begin{align*}
 \mathcal M_{D}  &  = \left\{\sum_{i=1}^n p_i\delta_{x_i}:n\in \N,~x_i\in\R\mbox{ and }p_i\in[0,1]~{\rm for~}i\in[n],~\sum_{i=1}^n p_i=1\right\} ;\\
 \mathcal M_{D,2}  &   =\left\{p\delta_x+(1-p)\delta_y:~x,y\in\R,~p\in[0,1]\right\}.
\end{align*}
Consider a function $f: A\to\R\cup\{-\infty,\infty\}$, where $A$ is a subset of $\R$. We say $f$ satisfies the lower semicontinuity if 
$\{x_n\}_{n\in\N}\subseteq A$ converges to $x\in A$ implies 
$\liminf_{n\to \infty} f(x_n)\ge f(x)$. Conversely,  $f$ is upper semicontinuous if $-f$ is lower semicontinuous.

For $F,G\in\mathcal M_C$, we write $F\preceq_1 G$ if $F(x)\ge G(x)$ for all $x\in\R$, that is, $G$ dominates $F$ in FSD. 
 For $F,G\in\mathcal M_C$, denote by $F\vee G$ the supremum of the set $\{F,G\}$ with respect to FSD. Such supremum always exists and is unique because the ordered space $(\mathcal M_C,\preceq_1)$ is a lattice; see e.g., \cite{KR00}, \cite{MS06} and \cite{MWW22}. Moreover, it can be explicitly described by $F\vee G(x)=\min\{F(x),G(x)\}$ for $x\in \R$.

We propose  the following axioms for $\rho: \mathcal M_C\to \R$. 

\renewcommand\theaxiom{ND}
\begin{axiom}[Nondegeneracy]
	\label{ax:ND}
	For $x,y\in\R$, $x\ne y$ implies $\rho(\delta_x)\ne \rho(\delta_y)$.
\end{axiom}

\renewcommand\theaxiom{LS}
\begin{axiom}[Lower semicontinuity]
	\label{ax:LS}
	For $\{F_n\}_{n\in\N}\subseteq\mathcal M_C$ and $F\in\mathcal M_C$, if $F_n\lawto F$, where $\lawto$ denotes the weak convergence, then $\liminf_{n\to\infty}\rho(F_n)\ge \rho(F)$.
\end{axiom}

\renewcommand\theaxiom{MaxS}
\begin{axiom}[Max-stability]
\label{ax:E}
For $F,G\in\mathcal M_C$, $\rho(F\vee G)=\rho(F)\vee \rho(G)$.
\end{axiom}

Clearly, Axioms \ref{ax:ND} and \ref{ax:LS}  are very weak and provide some regularity for $\rho$.
The main axiom of interest is Axiom \ref{ax:E}.

\textcolor{black}{
The financial interpretation of Axiom \ref{ax:E} can be explained by interpreting $\rho$  as a risk measure. 
A risk measure $\rho$ induces an acceptance set $\mathcal A$, 
which is the set of random variables with distribution $F$ satisfying $\rho(F)\le 0$. The set $\mathcal A$ is commonly interpreted as the set of random losses considered as acceptable without additional capital injection.  
Take any two  comonotonic random variables  $X$ and $Y$ representing financial losses.\footnote{Two real-valued functions $f,g$ on the same domain are comonotonic if $(f(x)-f(y))(g(x)- g(y))\ge 0$ for all $x,y$ in their domain. An equivalent condition is that they are both increasing transforms of the same function.} For instance, they can be the loss from the price change of an asset  and that of a call option on the same asset in a Black-Scholes model. 
For such $X$ and $Y$
 Axiom \ref{ax:E} is equivalent to the following property: if both $X$ and $Y$ are acceptable according to $\rho$, then $\max (X,Y)$ is acceptable; this is because the distribution of $\max(X,Y)$  is the supremum of the distributions of $X$ and $Y$ with respect to FSD. 
 Thus, the worst of two acceptable comonotonic losses is still acceptable. 
 This property is satisfied by the popular risk measure in finance, the Value-at-Risk; see Section \ref{sec:example}.}

Next, we say that $\rho$ is \emph{consistent with $\preceq_1$} if $\rho(F)\le \rho(G)$ whenever $F\preceq_1 G$. 
Notice that $\rho(G)=\rho(F\vee G)=\rho(F)\vee\rho(G)$ if $F\preceq_1 G$, and therefore Axiom \ref{ax:E} is stronger than consistency with $\preceq_1$.
Together with Axiom \ref{ax:ND}, this implies that $\rho$ satisfies the following property:
\renewcommand\theproperty{M}
\begin{property}[Monotonicity on constants]
	\label{ax:MC}
	For $x,y\in\R$, $x< y$ implies $\rho(\delta_x)<\rho(\delta_y)$.
\end{property}  
Property \ref{ax:MC} is useful in the proof of our main results. 

\subsection{Characterization}\label{sec:characterization}

Next, we present the main characterization result.
\textcolor{black}{Throughout, the terms  ``increasing" and ``decreasing" are in the weak sense.} 
\begin{theorem}\label{th-main}
Let $\rho:\mathcal M_C\to\R$. The following statements are equivalent.
\begin{itemize}
\item[(i)] Axioms \ref{ax:ND}, \ref{ax:LS} and \ref{ax:E} hold.
\item[(ii)] The representation holds
\begin{align}\label{eq-main}
\rho(F)=\sup_{x\in\R}\psi(x,F(x)) \mbox{ for all $F\in\mathcal M_C$},
\end{align}
for some
  function $\psi:\R\times[0,1]\to \R\cup\{-\infty\}$ that is decreasing and lower semicontinuous in the second argument and satisfies  $\psi(x,0)<\psi(y,0)$ if $x,y\in\R$ with $x<y$, and $\psi(x,1)=-\infty$ for all $x\in\R$.

\item[(iii)] The representation \eqref{eq-main} holds for some
  function $\psi:\R\times[0,1]\to \R\cup\{-\infty\}$ that is increasing and lower semicontinuous in the first argument  and decreasing and lower semicontinuous in the second argument and satisfies  $\psi(x,0)<\psi(y,0)$ if $x,y\in\R$ with $x<y$, and $\psi(x,1)=-\infty$ for all $x\in\R$.
\end{itemize}
\end{theorem}

The implications (iii) $\Rightarrow$ (i) and (ii) $\Rightarrow$ (iii) of Theorem \ref{th-main} are straightforward to verify, and their proofs are provided below.

{\bf Proof of (iii) $\Rightarrow$ (i):}
Suppose that $\rho$ admits the representation in the statement (iii). Axiom \ref{ax:ND} is trivial as one can check that $\rho(x)=\lim_{y\uparrow x}\psi(y,0)=\psi(x,0)$.  For $F,G\in\mathcal M_C$, we have
\begin{align*}
\rho(F\vee G)&=\sup_{x\in\R}\psi(x,F\vee G(x))
=\sup_{x\in\R}\psi(x,\min\{F(x),G(x)\})\\
&=\sup_{x\in\R}\left\{\psi(x,F(x))\vee \psi(x,G(x))\right\}
=\sup_{x\in\R}\psi(x,F(x))\vee \sup_{x\in\R}\psi(x,G(x))
=\rho(F)\vee \rho(G),
\end{align*}
where the third equality holds because $\psi$ is decreasing in the second argument. This gives Axiom \ref{ax:E}. 
To see Axiom \ref{ax:LS}, let $\{F_n\}_{n\in\N}\subseteq \mathcal M_C$ be a sequence such that $F_n\lawto F$ with $F\in\mathcal M_C$. It holds that
\begin{align*}
\liminf_{n\to\infty} \rho(F_n)
=\liminf_{n\to\infty} \sup_{x\in\R}\psi(x,F_n(x))
\ge \sup_{x\in\R}\liminf_{n\to\infty}\psi(x,F_n(x))
\ge \sup_{x\in\R}\psi(x,F(x))
=\rho(F),
\end{align*}
where the second inequality holds because $\psi$ is decreasing and lower semicontinuous in the second argument, and $\limsup_{n\to\infty} F_n(x)\le F(x)$ for any $x\in\R$. Hence, we have verified Axiom \ref{ax:LS}.

{\bf Proof of (ii) $\Rightarrow$ (iii):}
Denote by $\Psi_1$ and $\Psi_2$ the sets of the functions $\psi$ that satisfy all conditions in (ii) and (iii), respectively. 
For $\psi\in\Psi_1$, we  define $\rho_{\psi}(F)=\sup_{x\in\R}\psi(x,F(x))$ for $F\in\mathcal M_C$.
To prove this implication, it suffices to verify that for any $\psi\in\Psi_1$, there exists $\widetilde{\psi}\in\Psi_2$ such that $\rho_{\widetilde{\psi}}=\rho_{\psi}$ on $\mathcal M_C$. We claim that this assertion holds by setting 
\begin{align}\label{eq-deLS}
\widetilde{\psi}(x,p)=\sup_{t<x}\psi(t,p).
\end{align}
Indeed, for $F\in\mathcal M_C$, we have
\begin{align*}
\rho_{\widetilde{\psi}}(F)&=\sup_{x\in\R}\widetilde{\psi}(x,F(x))
=\sup_{x\in\R}\sup_{t<x}\psi(t,F(x))
=\sup_{t\in\R}\sup_{x>t}\psi(t,F(x))
=\sup_{t\in\R}\psi(t,F(t))=\rho_{\psi}(F),
\end{align*}
where the fourth step holds because $F$ is right continuous and  $\psi$ is decreasing and lower semicontinuous in the second argument.
It remains to show that $\widetilde{\psi}\in\Psi_2$. By the definition of $\widetilde{\psi}$, one can easily check that it is increasing and lower semicontinuous in the first argument. Notice that $p\mapsto\psi(t,p)$ is decreasing and lower semicontinuous for any $t\in\R$, and $\widetilde{\psi}$ is the supremum of these functions. This implies that $\widetilde{\psi}$ is decreasing and lower semicontinuous in the second argument. In addition, for $x<y$, we have
\begin{align*}
\widetilde{\psi}(x,0)=\sup_{t<x}\psi(t,0)<\psi\left(\frac{x+y}{2},0\right)\le \sup_{t<y}\psi(t,0)=\widetilde{\psi}(y,0),
\end{align*}
where the inequalities follows from the strictly increasing monotonicity of $t\mapsto \psi(t,0)$. {\color{black} It is also straightforward to verify that  $\psi(x, 1)=-\infty$ for all $x\in\R$ implies $\widetilde{\psi}(x, 1)=-\infty$ for all $x\in\R$.}
Hence, we have concluded that $\widetilde{\psi}\in\Psi_2$, and the implication (ii) $\Rightarrow$ (iii) holds.\qed

The most challenging part of Theorem \ref{th-main} is the implication (i) $\Rightarrow$ (ii).
We present the sketch of the proof for (i) $\Rightarrow$ (ii) below, and the detailed proof is put in Section \ref{sec:mainproof}. In the first step, we consider the situation that $\rho$ is restricted in $\mathcal M_{D,2}$. We will construct such a function $\psi$ in the theorem and verify that the representation holds on $\mathcal M_{D,2}$. The second step is to extend the representation to $\mathcal M_{D}$ by using Axiom \ref{ax:E} and the fact that any distribution in $\mathcal M_{D}$ can be obtained by finitely many joins acting on $\mathcal M_{D,2}$ with respect to $\preceq_1$. The last step is the extension from $\mathcal M_D$ to $\mathcal M_C$ by a standard convergence argument. 


\textcolor{black}{Examples of risk measures in finance with the form \eqref{eq-main} will be discussed in the next section. A simple example of $\psi$ satisfying the conditions in Theorem \ref{th-main} is $\psi(x,p)= x-1/(1-p)$, which belongs to the class in Example \ref{ex-lossVaR} below.  
We next provide an example that satisfies Axiom \ref{ax:E} but not Axiom \ref{ax:ND}.}

\begin{example}[Decreasing functions of the cdf] 
\label{ex:F}
Let $\psi(x,p)=g(p)-\infty\id_{\{x\neq x_0\}}$ for some $x_0\in\R$ and decreasing function $g:[0,1]\to\R$. Using representation \eqref{eq-main}, we have $\rho(F)=g(F(x_0))$ for $F\in\mathcal M_C$. In this case, $\rho$ satisfies Axiom \ref{ax:E}. Note that $\psi(x,0)=\psi(y,0)=-\infty$ whenever $x,y\neq x_0$, 
and thus, $\rho$ does not satisfy Axiom \ref{ax:ND}. 
Moreover, one can check that $\rho$ satisfies Axiom \ref{ax:LS} if and only if $g$ is lower semicontinuous.
\end{example}

Next, we present an analogue to Theorem \ref{th-main} by changing the maximum to a minimum.
 For $F,G\in\mathcal M_C$, denote by $F\wedge G$ the infimum of $\{F,G\}$ with respect to FSD. Such an infimum can be explicitly given by $F\wedge G(x)=\max\{F(x),G(x)\}$ for $x\in\R$. The following axiom for $\rho:\mathcal M_C\to\R$ is a dual property of max-stability. 

\renewcommand\theaxiom{MinS}
\begin{axiom}[Min-stability]
	\label{ax:MS}
	For $F,G\in\mathcal M_C$, $\rho(F\wedge G)=\rho(F)\wedge \rho(G)$.
\end{axiom}

To obtain a parallel result of Theorem \ref{th-main}, we require upper semicontinuity instead of lower semicontinuity.

\renewcommand\theaxiom{US}
\begin{axiom}[Upper semicontinuity]\label{ax:US}
For $\{F_n\}_{n\in\N}\subseteq\mathcal M_C$ and $F\in\mathcal M_C$, if $F_n\lawto F$, then $\limsup_{n\to\infty}\rho(F_n)\le \rho(F)$.
\end{axiom}

 \begin{proposition}\label{prop-mainMinS}
  	Let $\rho:\mathcal M_C\to\R$. The following statements are equivalent.
  	\begin{itemize}
  		\item[(i)]  Axioms \ref{ax:ND}, \ref{ax:US} and \ref{ax:MS} hold.
  		\item[(ii)] Denote by $F(x-)=\lim_{y\uparrow x}F(y)$
         the left continuous version of $F$.
        The representation holds
  		\begin{align}\label{eq-mainMinS}
  			\rho(F)=\inf_{x\in\R}\phi(x,F(x-)) \mbox{ for all $F\in\mathcal M_C$},
  		\end{align}
  		for some
  		function $\phi:\R\times[0,1]\to \R\cup\{\infty\}$ that is decreasing and lower semicontinuous in the second argument and satisfies  $\phi(x,1)<\phi(y,1)$ if $x,y\in\R$ with $x<y$, and $\phi(x,0)=\infty$ for all $x\in\R$.
  		
  		\item[(iii)] The representation \eqref{eq-mainMinS} holds for some
  		function $\phi:\R\times[0,1]\to \R\cup\{\infty\}$ that is increasing and lower semicontinuous in the first argument and decreasing and lower semicontinuous in the second argument and satisfies  $\phi(x,1)<\phi(y,1)$ if $x,y\in\R$ with $x<y$, and $\phi(x,0)=\infty$ for all $x\in\R$.
  	\end{itemize}
  \end{proposition}
  
We omit the proof of the above proposition as it is symmetric to that of Theorem \ref{th-main}.

\begin{remark}
    {\color{black} \cite{KZ24} studied functionals satisfying max-stability  under the name of maxitivity in parallel to our work and obtained a representation similar to Theorem \ref{th-main}.  
Their main results cover a more general setting, which includes, for instance,  the property of  max-stability
with respect to SSD. 
Their  proof techniques are quite different from ours, as explained below.  \cite{KZ24} formulated  max-stability   in terms of quantile functions. For $\rho:\mathcal M_C\to \R$,
define $\widetilde{\rho}$ as $\widetilde{\rho}(Q)=\rho(Q^{-1+})$ for $Q\in\mathcal Q_C$, where $\mathcal Q_C$ denotes the set of quantile functions associated with $\mathcal M_C$, and $Q^{-1+}$ is the right-inverse function of $Q$. Their proof focuses on reformulating the sublevel sets $\mathcal A_s=\{Q: \widetilde{\rho}(Q)\le s\}$ for $s\in\R$.
To this end, they introduced a bivariate functional $J: \R \times \mathcal L \to \R$, where $\mathcal L$ is the set of all lower semicontinuous functions on $\R$, and showed that each sublevel set can be expressed as $\mathcal A_s=\{Q: J(s,Q(s))\le 0\}$. The properties of $J$ in its second argument allow for a dual representation of $J$. Together with the identity $\widetilde{\rho}(Q)=\inf\{s\in\R: Q\in\mathcal A_s\}$ and some further standard calculations, they derived a representation of $\widetilde{\rho}$ analogous to \eqref{eq-main}.
In contrast, our proof of Theorem \ref{th-main} takes a more direct approach. We construct the functional $\psi$ in \eqref{eq-main} by analyzing two-point distributions, and then verify that the resulting representation extends to all distributions in $\mathcal M_C$. 
}
\end{remark}

\section{Three examples of risk measures in finance}\label{sec:example}

In this section, we present three examples of functionals  that can be expressed in the form of \eqref{eq-main}, with the specific forms of $\psi$ provided.

These three examples are all used as risk measures in finance. 
The first example is the quantile at a fixed probability level, known as the Value-at-Risk (VaR), one of the most prominent risk measures in financial regulation (see \cite{MFE15}). 

The next two examples generalize VaR in different ways. 
 The benchmark-loss VaR introduced by \cite{BBM2020}  generalizes a usual quantile  by considering the supremum of quantile over different probability levels  minus a benchmark quantile function. 
 The    $\Lambda$-quantile introduced by \cite{FMP14} 
 generalizes a usual quantile by replacing the probability level with a curve  $x\mapsto \Lambda(x)$.

\begin{example}[VaR]\label{ex-VaR}
Let $\psi(x,p)=x-\infty\id_{\{p\ge \alpha\}}$ for some $\alpha\in(0,1)$. In this case, $\rho(F)=\sup\{x: F(x)<\alpha\}$. This is the left-quantile at level $\alpha$.
\end{example}

\begin{example}[Benchmark-loss VaR]\label{ex-lossVaR}
Let $\psi(x,p)=x-h(p)$, where $h:[0,1]\to\R\cup\{\infty\}$ is increasing and upper semicontinuous and \textcolor{black}{satisfies $h(1)=\infty$}. In this case, we have
\begin{align*}
\rho(F)=\sup_{x\in\R}\{x-h(F(x))\},~~F\in\mathcal M_C.
\end{align*}
Denote the $\alpha$-left quantile function of the distribution $F$ by $F^{-1}(\alpha)$, as represented in Example \ref{ex-VaR}.
By standard manipulation,
the functional $\rho$ can be represented based on the usual quantiles as 
\begin{align*}
\rho(F)=\sup_{\alpha\in[0,1]}\{F^{-1}(\alpha) -h(\alpha)\},~~F\in\mathcal M_C.
\end{align*} 
This is exactly the benchmark-loss VaR.
In particular, if $h(p)=\infty\id_{\{p\ge \alpha_0\}}$ for some $\alpha_0\in(0,1)$,
then $\rho$ reduces to the left-quantile at level $\alpha_0$. 

\end{example}

\begin{example}[$\Lambda$-quantile]\label{ex-LambdaVaR}
Let $\psi(x,p)=x-\infty\id_{\{p\ge \Lambda(x)\}}$, where $\Lambda$ is a decreasing function from $\R$ to $[0,1]$. In this case, we have
\begin{align*}
\rho(F)=\sup\{x: F(x)<\Lambda(x)\},~~F\in\mathcal M_C.
\end{align*}
This is exactly the $\Lambda$-quantile.
While $\psi$ may not be increasing or lower semicontinuous in the first argument, we can define $\widetilde{\psi}$ as the form in \eqref{eq-deLS}, i.e.,
\begin{align*}
\widetilde{\psi}(x,p)=\sup_{t<x}\psi(t,p)=\sup\{t\in(-\infty,x):\Lambda(t)\le p\},
\end{align*}
By the arguments in the proof of Theorem \ref{th-main} (ii) $\Rightarrow$ (iii) in Section \ref{sec:characterization}, $\widetilde{\psi}$ is increasing and lower semicontinuous in the first argument, and
$
\rho(F)=\sup_{x\in\R}\widetilde{\psi}(x,F(x))
$
for all $F\in\mathcal M_C$.

Theorem 3.1 of \cite{HWWX21} gave a representation of $\Lambda$-quantiles based on the usual quantiles:
\begin{align}\label{eq-LambdaVaR-VaR}
\rho(F)=\sup_{x\in\R}\{F^{-1}(\Lambda(x))\wedge x\}
=\inf_{x\in\R}\{F^{-1}(\Lambda(x))\vee x\},~~F\in\mathcal M_C.
\end{align}
In particular, if $\Lambda(x)=\alpha_0$ for all $x\in\R$ with some $\alpha_0\in(0,1)$, then $\rho$ reduces to the left-quantile at level $\alpha_0$.

\end{example}

For a nonempty set $\mathcal X$, the \emph{super level set} of a functional $f:\mathcal X\to\R$ is defined as $\{x\in\mathcal X: f(x)\ge t\}$ for $t\in\R$.
Figure \ref{fig-SLS} shows the super level set of $\psi$ in Examples
\ref{ex-VaR}, \ref{ex-lossVaR} and \ref{ex-LambdaVaR}. For any threshold $t\in\R$, 
the super level set is the region below a curve.
In particular, the super level sets of $\Lambda$-quantile are L-shaped.

\usetikzlibrary{intersections} 
\usepgfplotslibrary{fillbetween}
\usetikzlibrary{patterns}
\pgfplotsset{compat=1.16}
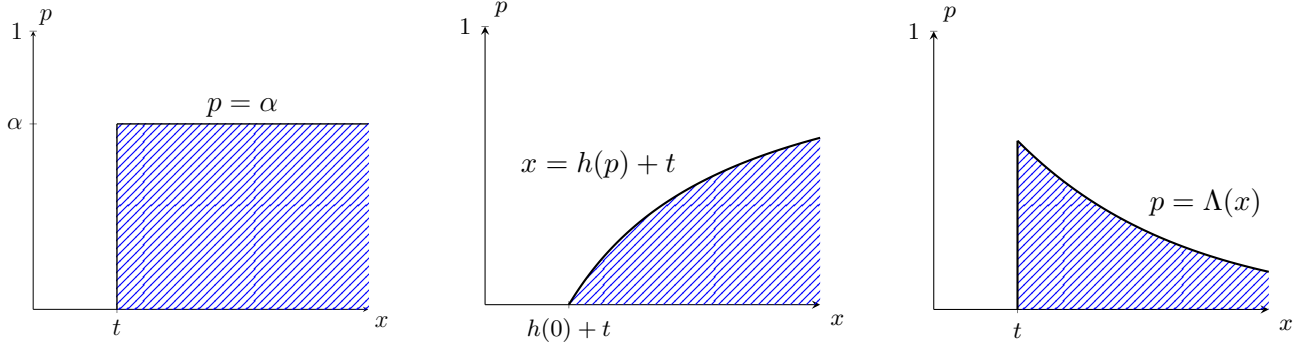
\begin{figure}[t]
\centering
\begin{minipage}{0.28\textwidth}
\begin{tikzpicture}[scale=0.65]
\begin{axis}[
    xlabel={$x$},
    ylabel={$p$},
    xmin=0, xmax=2,
    ymin=0, ymax=1,
    axis lines=middle,
    xlabel style={at={(ticklabel* cs:1)},anchor=north west},
    ylabel style={at={(ticklabel* cs:1)},anchor=south west},
    xtick={0.5},
    xticklabels={$t$},
    ytick={0,2/3,1},
    yticklabels={0,$\alpha$,1},
    tick label style={font=\Large},
    label style={font=\Large},
    clip=false,
]
\path[pattern=north east lines, pattern color=black] (axis cs:0.5,0) -- (axis cs:0.5,2/3) -- (axis cs:2,2/3) -- (axis cs:2,0) -- cycle;
\draw[thick] (axis cs:0.5,0) -- (axis cs:0.5,2/3);
\draw[thick] (axis cs:0.5,2/3) -- (axis cs:2,2/3)
node[pos=0.5,anchor=south] {\LARGE $p=\alpha$};
\end{axis}
\end{tikzpicture}
\end{minipage}
\hfill
\begin{minipage}{0.28\textwidth}
\begin{tikzpicture}[scale=0.65]
\begin{axis}[
    xlabel={$x$},
    ylabel={$p$},
    xmin=0, xmax=2,
    ymin=0, ymax=1,
    axis lines=middle,
    xlabel style={at={(ticklabel* cs:1)},anchor=north west},
    ylabel style={at={(ticklabel* cs:1)},anchor=south west},
    xtick={0.5},
    xticklabels={$h(0)+t$},
    ytick={0,1},
    yticklabels={0,1},
    tick label style={font=\footnotesize},
    label style={font=\footnotesize},
    domain=0.5:2,
    samples=100,
    clip=false,
]
\addplot+[mark=none, name path=f, thick, color=black] {1-1/(x+0.5)} node[pos=0.5,anchor=south east] {$x=h(p)+t$};
\path[name path=axis] (axis cs:0.5,0) -- (axis cs:2,0);
\addplot+[black!50, pattern=north east lines, pattern color=black] fill between[of=f and axis, soft clip={domain=0.5:2}];
\end{axis}
\end{tikzpicture}
\end{minipage}
\hfill
\begin{minipage}{0.28\textwidth}
\begin{tikzpicture}[scale=0.65]
\begin{axis}[
    xlabel={$x$},
    ylabel={$p$},
    xmin=0, xmax=2,
    ymin=0, ymax=1,
    axis lines=middle,
    xlabel style={at={(ticklabel* cs:1)},anchor=north west},
    ylabel style={at={(ticklabel* cs:1)},anchor=south west},
    xtick={0.5},
    xticklabels={$t$},
    ytick={0,1},
    yticklabels={0,1},
    tick label style={font=\footnotesize},
    label style={font=\footnotesize},
    domain=0.5:2,
    samples=100,
    clip=false,
]
\addplot+[mark=none, name path=f, thick, color=black] {exp(-x)} node[pos=0.5,anchor=south west] {$p=\Lambda(x)$};
\path[name path=axis] (axis cs:0.5,0) -- (axis cs:2,0);
\addplot+[black!50, pattern=north east lines, pattern color=black] fill between[of=f and axis, soft clip={domain=0.5:2}];
\draw[thick] (axis cs:0.5,0) -- (axis cs:0.5,0.6065);
 \end{axis}
 

\end{tikzpicture}
\end{minipage}
\caption{The super level set of $\psi$ that induces VaR, benchmark-loss VaR and $\Lambda$-quantile. Left: VaR ($\psi(x,p)=x-\infty\id_{\{p\ge \alpha\}}$); Middle: benchmark-loss VaR ($\psi(x,p)=x-h(p)$); Right: $\Lambda$-quantile ($\psi(x,p)=x-\infty\id_{\{p\ge \Lambda(x)\}}$).}\label{fig-SLS}
\end{figure}

\begin{remark}
\color{black} Another popular risk measure used in finance is the Expected Shortfall (ES), defined as
$$\rho(F)= \frac{1}{1-\alpha}\int_\alpha^1 F^{-1}(\beta)\d \beta ,$$
where $\alpha \in(0,1)$ is a parameter.
It is straightforward to check that ES for any $\alpha\in(0,1)$ does not satisfy either Axiom \ref{ax:E} or Axiom \ref{ax:MS}. Nevertheless, ES satisfies  a weak notion of max-stability with respect to SSD over a convex set of distributions, as shown by \citet[Theorem 1]{MWW22}. See  \cite{KZ24} for characterizing  max-stability with respect to SSD. 
\end{remark}



\section{A characterization of \texorpdfstring{$\Lambda$}{Lambda}-quantile}
\label{sec:lambda}

\subsection{Representation}
With max-stability and min-stability characterized in Theorem \ref{th-main} and Proposition \ref{prop-mainMinS}, 
one may wonder which mappings satisfy both conditions. Clearly, 
for a decreasing function $g:[0,1]\to \R$ and $x_0\in \R$ 
the mapping $\rho:F\mapsto g(F(x_0))$
satisfies both Axioms \ref{ax:E} and \ref{ax:MS}, which we saw in Example \ref{ex:F}. 
This mapping does not satisfy \ref{ax:ND}.

It turns out that the $\Lambda$-quantiles in Example \ref{ex-LambdaVaR} are the unique class of functionals that satisfy Axioms \ref{ax:ND}, \ref{ax:LS}, \ref{ax:E} and \ref{ax:MS}.

\begin{theorem}
\label{th:main-MM}
For $\rho:\mathcal M_C\to\R$,   Axioms \ref{ax:ND}, \ref{ax:LS}, \ref{ax:E} and \ref{ax:MS} hold if and only if there exists 
a strictly increasing function $f:\R\to\R$ and 
a decreasing function $\Lambda:\R\to(0,1]$ such that $\rho(F)=\sup\{f(x):F(x)<\Lambda(x)\}$ for all $F\in\mathcal M_C$.
 \end{theorem}

\textcolor{black}{The proof of Theorem \ref{th:main-MM} is   in Section \ref{sec:pf2}.}



The class of $\Lambda$-quantiles was characterized by \cite{BP22} with several properties. Among them, the key property is called locality: for any interval $(x,y)\subseteq \R$, $\rho(F)=\rho(G)$ if $\rho(F)\in(x,y)$ and $F=G$ on $(x,y)$. 
This property implies that the value of the functional $\rho(F)$ depends only on the values of the distribution function $F$ within arbitrarily small open neighborhoods 
around the functional's value.
\textcolor{black}{To compare the financial interpretations for a risk measure $\rho$ used by a regulator, locality means that the regulator is insensitive to any changes in the loss distribution function $F$ that is outside a neighborhood of  $\rho(F)$. 
Axioms \ref{ax:E} and \ref{ax:MS} mean that for any two comonotonic loss random variables $X$ and $Y$, if both are acceptable, then so is their maximum; if both are  unacceptable, then so is their minimum. 
Thus, the financial interpretations are quite different, although mathematically they lead to the same class of functionals.} 

\subsection{Connection to the result of \cite{CGR08}} 

\citet[Theorem 18]{CGR08} gave a characterization of functionals 
via Sugeno integrals on lattices of general functions taking values in $[0,1]$ (not necessarily distribution functions) satisfying max-stability and min-stability for comonotonic functions, as well as some other standard properties. 
Denote by $\mathcal M_{[0,1]}$ the set of all distributions with support in $[0,1]$.
Our Theorem \ref{th:main-MM} is closely related to that result if we consider all distributions in $\mathcal M_{[0,1]}$, but it is unclear how their result can be generalized to $\mathcal M_C$.

It takes some effort to see the connection under the setting of  $\mathcal M_{[0,1]}$, which we explain below. 
Let $\rho:\mathcal M_{[0,1]}\to\R$, and suppose that $\mathcal Q_{[0,1]}$ is the set of all right-quantiles with  range in $[0,1]$. 
Note that 
\begin{align}\label{eq-MQeq}
\mathcal M_{[0,1]}=\{Q^{-1+}: Q\in\mathcal Q_{[0,1]}\},
\end{align}
where $Q^{-1+}$ is the right-inverse function of $Q$.
Due to this one-to-one correspondence between $\mathcal M_{[0,1]}$ and $\mathcal Q_{[0,1]}$, we can define a mapping 
$\widetilde{\rho}:\mathcal Q_{[0,1]}\to\R$ as $\widetilde{\rho}(Q)=\rho(Q^{-1+})$ for all $Q\in\mathcal Q_{[0,1]}$.
Based on this,
our axioms on $\rho$ can be equivalently reformulated in terms of right-quantiles on $\widetilde{\rho}$. In particular, max-stability and min-stability for $\widetilde{\rho}$ is defined as $\widetilde{\rho}(Q_1\vee Q_2)=\widetilde{\rho}(Q_1)\vee \widetilde{\rho}(Q_2)$ and $\widetilde{\rho}(Q_1\wedge Q_2)=\widetilde{\rho}(Q_1)\wedge \widetilde{\rho}(Q_2)$, respectively, for all $Q_1,Q_2\in\mathcal Q_{[0,1]}$, where on the space of quantiles, the maximum and the minimum are understood pointwise. To obtain the representation of $\rho$, it is equivalent to consider the representation of $\widetilde{\rho}$ instead. 
 Since right-quantiles are all increasing, implying that any pair of them are comonotonic,  
the representation of $\widetilde{\rho}$ satisfying max-stability, min-stability, lower semicontinuous and $\widetilde{\rho}(x)=x$ for all $x\in[0,1]$
can be obtained from \citet[Theorem 18]{CGR08}:
\begin{align}\label{eq-CGR1}
\widetilde{\rho}(Q)=\sup_{x\in[0,1]}\{x\wedge v(\{\alpha:Q(\alpha)\ge x\})\},~~~Q\in\mathcal Q_{[0,1]},
\end{align}
where $v$ is an inner-continuous normalized capacity on $([0,1],\mathcal B([0,1]))$.\footnote{Let $(\Omega,\mathcal F)$ be a measurable space. A set function $v:\mathcal F\to [0,1]$ is a normalized capacity if $v(\emptyset)=1-v(\Omega)=0$ and $v(A)\le v(B)$ whenever $A\subseteq B$. We say $v$ is inner-continuous if $A_n\subseteq A_{n+1}$ for $n\in\N$ and $\bigcup_{n=1}^\infty A_n=A$ imply
$\lim_{n\to\infty} v(A_n)=v(A)$.} 
Recall the relation in \eqref{eq-MQeq} and note that
$Q(\alpha)\ge x$ is equivalent to $\alpha\ge Q^{-1+}(x)$ for $\alpha,x\in[0,1]$.
The representation of $\rho$ can be derived from \eqref{eq-CGR1}:
\begin{align*}
\rho(F)=\sup_{x\in[0,1]} \{x\wedge v(\{\alpha: \alpha\ge F(x)\})\}
=\sup_{x\in[0,1]} \{x\wedge v([F(x),1])\},~~~F\in\mathcal M_{[0,1]}.
\end{align*}
Denote by $f(p)=v([p,1])$ for $p\in[0,1]$. The above formulation becomes
\begin{align*}
\rho(F)
=\sup_{x\in[0,1]} \{x\wedge f(F(x))\},~~~F\in\mathcal M_{[0,1]}.
\end{align*}
By the property of inner-continuous normalized capacity, we have that $f:[0,1]\to[0,1]$ is decreasing and right-continuous with $f(1)=1-f(0)=0$. Note that $F$ is increasing and right-continuous, which further implies that $x\mapsto f(F(x))$ is decreasing and right-continuous. Hence, it is straightforward to check that 
\begin{align*}
\rho(F)=\sup\{x: x<f(F(x))\}=\sup\{x: F(x)<f^{-1}(x)\},~~~F\in\mathcal M_{[0,1]}.
\end{align*}
where $f^{-1}$ can be any inverse function of $f$. This yields the representation of a $\Lambda$-quantile. Their proof relies on the fact that any finitely supported distribution in $\mathcal M_{[0,1]}$  can be obtained by  repeated applications of maxima and minima of two-point distributions on $\{0,1\}$ and degenerate distributions at $x\in [0,1]$.
This argument requires that all functionals take values within a compact interval, and it is unclear how it can be generalized to   $\R$. In contrast, our proof differs significantly and is capable of handling $\mathcal M_C$.


\section{Conclusion}\label{sec:con}

We provide a characterization of functionals satisfying three axioms of monotonicity, lower semicontinuity, and max-stability. 
Such functionals include the risk measures,  $\Lambda$-quantiles and benchmark-loss VaR  as special cases. 
With max-stability replaced by  min-stability, a parallel characterization result is also established. Additionally, we demonstrate that $\Lambda$-quantiles (possibly preceding by strictly increasing function) are the only cases that satisfy both max-stability and min-stability.

We focus on max-stability and {\color{black}min}-stability with respect to FSD in this paper, but this property can be generally defined on other spaces. For instance, 
in financial mathematics, a risk measure is a functional that takes random variables as input (\cite{FS16}). 
In this setting, max-stability can be defined as $\rho(X\vee Y) = \rho(X) \vee \rho(Y)$, as studied by \cite{KZ21},  where $X\vee Y$ represents the pointwise maximum of $X$ and $Y$. 
\cite{KZ21} showed that in the above setting, decision makers who adopt max-stable risk measures base their assessment on evaluation under the worst-case scenarios. 
Additionally, \cite{K79} considered complete and transitive binary preference relations over sets. The property stated in Equation (1.2) of \cite{K79} 
characterizes all preference relations that are governed by maximization of a preference over the elements.
\textcolor{black}{Max-stability with respect to some other partial orders  is studied by \cite{KZ24}.}

\section{Proofs}\label{sec:mainproof}

\subsection{Proof of (i) $\Rightarrow$ (ii) of Theorem \ref{th-main}}
In this section, we assume that $\rho$ satisfies Axioms \ref{ax:ND}, \ref{ax:LS} and \ref{ax:E}. Recall that Axiom \ref{ax:E} is stronger than consistency with $\preceq_1$, and 
together with  Axiom \ref{ax:ND}, $\rho$ satisfies Property \ref{ax:MC}.
We will frequently use these two properties in the proof.

\subsubsection{Proof on $\mathcal M_{D,2}$}
In this step, our main purpose is to construct a function $\psi$ that satisfies all conditions in Theorem \ref{th-main} (ii) such that $\rho(F)=\sup_{x\in\R}\psi(x,F(x))$ for all $F\in \mathcal M_{D,2}$. 



For $c\in\R$, we abuse the notation by setting $\rho(c)=\rho(\delta_c)$. It follows from consistency with $\preceq_1$ and Axiom \ref{ax:LS} that the mapping $c\mapsto\rho(c)$ is increasing and left continuous on $\R$.
Denote by $A=\{(x,y,p): x,y\in\R,~x\le y,~p\in[0,1]\}$. Define $f:A\to \R$ as a function such that 
\begin{align}\label{eq-f}
f(x,y,p)=\rho(p\delta_x+(1-p)\delta_y),~~(x,y,p)\in A.
\end{align}
 By consistency with $\preceq_1$ and Axiom \ref{ax:LS},
we have the following properties of $f$.
\begin{itemize}
	\item[(i)] $f(x,x,p)=\rho(x)$; 
	$f(x,y,0)=\rho(y)$; $f(x,y,1)=\rho(x)$.
	\item[(ii)] $f(x,y,p)$ is increasing and left continuous in $x$. 
	\item[(iii)] $f(x,y,p)$ is increasing and left continuous in $y$. 
	\item[(iv)] $f(x,y,p)$ is decreasing and right continuous in $p$. 
\end{itemize} 
Further, we use the convention that $\inf \emptyset=\infty$ and define
\begin{align}\label{eq-h}
	h(x,p)=\inf\{y\in[x,\infty):~f(x,y,p)>\rho(x)\},~~x\in\R,~p\in[0,1].
\end{align}
If $h(x,p)<\infty$, then we have
\begin{align}\label{eq:property1}
 f(x,y,p)=\rho(x) \mbox{~if~}  y\in[x,h(x,p)] \mbox{~~~and~~~} f(x,y,p)>\rho(x) \mbox{~if~}  y>h(x,p) ,
 \end{align}
because $f(x,x,p)=\rho(x)$ and $f(x,y,p)$ is increasing and left continuous in $y$.
It follows from Property \ref{ax:MC} that $h(x,0)=x$, and it is clear that $h(x,1)=\infty$.
The next result gives other properties of $h$.
Proofs of all propositions in this section are provided in Section \ref{sec:h}.
\begin{proposition}\label{prop:h}
Let $h$ be defined in \eqref{eq-h}. The following statements hold.
\begin{itemize}
	\item[(i)] $h(x,p)$ is increasing in $x$.
	\item[(ii)] $h(x,p)$ is increasing and upper semicontinuous in $p$.
\end{itemize}
\end{proposition}

For $x\in \R$, denote by
\begin{align}\label{eq-Ax}
 A_x=\{(y,p):y>h(x,p),~p\in[0,1]\}
\end{align}
 and define $H_x:A_x\to\R$ as 
\begin{align}\label{eq-H}
H_x(y,p)=f(x,y,p),~~ (y,p)\in A_x.
\end{align} 
It is straightforward to check that $(y,0)\in A_x$ for any $x,y\in\R$ with $x< y$, and $H_x(y,0)=\rho(y)$.

By the definitions of $f$, $h$ and $H_x$
  in \eqref{eq-f}, \eqref{eq-h} and \eqref{eq-H}, respectively, we can immediately obtain   
\begin{align}\label{eq-D2-1}
	f(x,y,p)=\rho(x)\id_{\{(y,p)\notin A_x\}}+H_x(y,p)\id_{\{(y,p)\in A_x\}},~~(x,y,p)\in A.
\end{align}

Further, we present the properties of the function defined in \eqref{eq-H}.
\begin{proposition}\label{prop:H}
Let $A_x$ and $H_x$ be defined in \eqref{eq-Ax} and \eqref{eq-H}, respectively.
The following statements hold.
\begin{itemize}
	\item[(i)] For fixed $x\in\R$ and $p\in[0,1]$, $H_x(y,p)$ is increasing and left continuous in $y$ on the interval $(h(x,p),\infty)$.
	
	\item[(ii)] For fixed $x,y\in\R$ with $x\le y$, $H_x(y,p)$ is decreasing and right continuous in $p$ on the interval $\{p\in[0,1]: y>h(x,p)\}$.
	
	\item[(iii)] For fixed $y\in\R$ and $p\in[0,1]$, $H_x(y,p)$ is a constant in $x$ on the interval $\{x\in \R: y>h(x,p)\}$.

	\item[(iv)] For $x_1,x_2\in\R$ with $x_1<x_2$, we have 
	$A_{x_2}\subseteq A_{x_1}$.
	If $(y,p)\in A_{x_1}\setminus A_{x_2}$, then $$H_{x_1}(y,p)=f(x_1,y,p)\le f(x_2,y,p)= \rho(x_2).$$
\end{itemize}
\end{proposition}

Finally, we define
\begin{align}\label{eq-psi}
\psi(y,p)=\sup_{x\in\R}\left\{H_x(y,p)\id_{\{(y,p)\in A_x\}}-\infty\id_{\{(y,p)\notin A_x\}}\right\},~~y\in\R,~p\in[0,1].
\end{align}
The main purpose is to verify that $\psi$ satisfies all conditions in Theorem \ref{th-main} and $\rho(F)=\sup_{y\in\R}\psi(y,F(y))$ for all $F\in\mathcal M_{D,2}$. We summarize these conclusions into the following result.

\begin{proposition}\label{prop:D2}
Let $f$ and
$\psi$ be defined in \eqref{eq-f} and \eqref{eq-psi}, respectively. The following statements hold.
\begin{itemize}
	\item[(i)] $\psi$ is
	 decreasing and lower semicontinuous in the second argument. 
	\item[(ii)]  $\psi(x,0)<\psi(y,0)$ for $x,y\in\R$ with $x<y$ and $\psi(x,1)=-\infty$ for any $x\in\R$.
	\item[(iii)] $\rho(F)=\sup_{y\in\R}\psi(y,F(y))$ for all $F\in\mathcal M_{D,2}$.
\end{itemize}
\end{proposition}

Proposition \ref{prop:D2} justifies the necessity statement in Theorem \ref{th-main} on $\mathcal M_{D,2}$.

\subsubsection{Proofs of Propositions \ref{prop:h}, \ref{prop:H} and \ref{prop:D2}}\label{sec:h}

\begin{proof}[Proof of Proposition \ref{prop:h}]
(i) We prove this statement by contradiction. For $p\in[0,1]$, assume that there exists $x_1,x_2\in\R$ such that $x_1<x_2$ and $a_1>a_2$, where $a_i=h(x_i,p)$ for $i=1,2$. Let $F=\delta_{x_2}$ and $G=p\delta_{x_1}+(1-p)\delta_{a_1}$. It is clear that $a_1>a_2>x_2>x_1$. Hence, we have $F\vee G=p\delta_{x_2}+(1-p)\delta_{a_1}$. Axiom \ref{ax:E} implies
\begin{align}
\label{eq:proof1}
	f(x_2,a_1,p)=\rho\left(F\vee G\right)=\rho(F)\vee \rho(G)=\rho(x_2) \vee f(x_1,a_1,p).
\end{align} 
By the definition of $a_1$ and \eqref{eq:property1}, it holds that $f(x_1,a_1,p)=\rho(x_1)$ and $f(x_2,a_1,p)>\rho(x_2)$. Thus, 
\begin{align*}
\rho(x_2)<f(x_2,a_1,p)=\rho(x_2) \vee f(x_1,a_1,p)
=\rho(x_2)\vee\rho(x_1)=\rho(x_2),
\end{align*}
where the first equality follows from \eqref{eq:proof1} and the last inequality is due to Property \ref{ax:MC}. 
This yields a contradiction.

(ii) Let $p_1,p_2\in[0,1]$ satisfying $p_1\le p_2$ and $x\in\R$. We have $f(x,y,p_1)\ge f(x,y,p_2)$ for any $y\ge x$. This implies $\{y\in[x,\infty): f(x,y,p_1)>\rho(x)\}\supseteq \{y\in[x,\infty): f(x,y,p_2)>\rho(x)\}$. It then follows that $h(x,p_1)\le h(x,p_2)$, and this gives that $h(x,p)$ is increasing in $p$.

To see the upper semicontinuity, we need to verify that $\lim\sup_{p\to p_0} h(x,p)\le h(x,p_0)$ holds for any fixed $x\in\R$ and $p_0\in[0,1]$. The case that $h(x,p_0)=\infty$ is trivial, and we assume that $h(x,p_0)<\infty$. Note that $h(x,p)$ is increasing in $p$. It suffices to prove that $p\mapsto h(x,p)$ is right continuous at $p_0$.
Let $a_0=h(x,p_0)$.
We assume by contradiction that there exists $\epsilon_0>0$ and $\{p_n\}_{n\in\N}\subseteq[0,1]$ with $p_n\downarrow p_0$ such that $a_n>a_0+\epsilon_0$, where $a_n=h(x,p_n)$. It follows from \eqref{eq:property1} that
\begin{align}\label{eq-prop2eq}
	\rho(x)=f(x,a_n,p_n)\ge f(x,a_0+\epsilon_0,p_n).
\end{align}
Since $f(x,y,p)$ is right continuous in $p$, we have
\begin{align*}
\rho(x)<f(x,a_0+\epsilon_0,p_0)=\lim_{n\to\infty}f(x,a_0+\epsilon_0,p_n)\le\rho(x),
\end{align*}
where the inequality follows from the relation $a_0=h(x,p_0)$ and the definition of $h$ in \eqref{eq-h},
and the last equality is due to \eqref{eq-prop2eq}. This yields a contradiction, and we complete the proof.
\end{proof}

\begin{proof}[Proof of Proposition \ref{prop:H}]

(i) and (ii) are easy to see because $H_x$ has the same value of the mapping $(y,p)\mapsto f(x,y,p)$ and $f$ satisfies the corresponding properties.

(iii) For $y\in\R$ and $p\in[0,1]$, let  
 $x_1,x_2\in\R$ satisfying $x_1\le x_2\le y$ and $y>h(x_i,p)$ for $i=1,2$. Denote by $F=\delta_{x_2}$ and $G=p\delta_{x_1}+(1-p)\delta_y$. It holds that $F\vee G=p\delta_{x_2}+(1-p)\delta_y$. Axiom \ref{ax:E} implies
\begin{align}\label{eq-eqD2}
	f(x_2,y,p)=\rho\left(F\vee G\right)=\rho(F)\vee \rho(G)={\color{black}\rho(x_2)} \vee f(x_1,y,p).
\end{align} 
Since $y>h(x_i,p)$ for $i=1,2$, we have $f(x_1,y,p)>{\color{black}\rho(x_2)}$. Combining with \eqref{eq-eqD2}, we obtain $f(x_2,y,p)=f(x_1,y,p)$. This completes the proof of (iii).

(iv) Proposition \ref{prop:h} gives the decreasing monotonicity of $A_x$ in $x$. By \eqref{eq-D2-1}, we have $f(x_1,y,p)=H_{x_1}(y,p)$ and $f(x_2,y,p)=\rho(x_2)$ as $(y,p)\in A_{x_1}\setminus A_{x_2}$. Since $x_1\le x_2$, consistency with $\preceq_1$ implies $f(x_1,y,p)\le f(x_2,y,p)$. This completes the proof. \end{proof}

\begin{proof}[{Proof of Proposition \ref{prop:D2}}]
Denote by $V_x(y,p)=H_x(y,p)\id_{\{(y,p)\in A_x\}}-\infty\id_{\{(y,p)\notin A_x\}}$ for $x,y\in[0,1]$ and $p\in[0,1]$, and we have $\psi(y,p)=\sup_{x\in\R}V_x(y,p)$.

(i) For any fixed $x,y\in\R$,
it suffices to show that $p\mapsto V_x(y,p)$ is decreasing and lower semicontinuous.
If $x\ge y$, then $(y,p)\notin A_x$ for any $p\in[0,1]$. In this case, we have $V_x(y,p)=-\infty$ for all $p\in[0,1]$. Assume now $x<y$,  
there exists $p_0\in(0,1]$ such that 
\begin{align*}
\{p\in[0,1]:(y,p)\in A_x\}=\{p\in[0,1]:y>h(x,p)\}=[0,p_0),
\end{align*}
where the last step holds because  $h(x,p)$ is increasing and upper continuous in $p$ (see Proposition \ref{prop:h} (ii)). Moreover, 
it follows from Proposition \ref{prop:H} (ii) that $p\mapsto H_x(y,p)$ is decreasing and right continuous on 
$[0,p_0)$. Therefore, one can check that $p\mapsto V_x(y,p)$ is decreasing and lower semicontinuous. This completes the proof of {\color{black}(i)}.

(ii) When $p=0$, we have
\begin{align}
\label{eq:proof2}
\psi(y,0)=\sup_{x\in\R} \left\{H_x(y,0)\id_{\{(y,0)\in A_x\}}-\infty\id_{\{(y,0)\notin A_x\}}\right\}
=\sup_{x\in\R}\left\{\rho(y)\id_{\{y>x\}}-\infty\id_{\{y\le x\}}\right\}=\rho(y),
\end{align}
where we have used the left continuity of $c\mapsto \rho(c)$ in the last step.
Property \ref{ax:MC} implies $\rho(x)=\psi(x,0)<\psi(y,0)=\rho(y)$ if $x<y$.

When $p=1$, noting that $h(x,1)=\infty$ for any $x\in\R$, we have
\begin{align}
\label{eq:proof3}
	\psi(y,1)=\sup_{x\in\R} \left\{H_x(y,1)\id_{\{(y,1)\in A_x\}}-\infty\id_{\{(y,1)\notin A_x\}}\right\}
	=-\infty.
\end{align}
This completes the proof of (ii).

(iii) Denote by $\rho'(F)=\sup_{z\in\R}\psi(z,F(z))$. 
For $x\le y$ and $p\in[0,1]$, we define  $g(x,y,p)=\rho'(p\delta_x+(1-p)\delta_y)$.
We aim to show that $g(x,y,p)=f(x,y,p)$ for any $x\le y$ and $p\in[0,1]$. Recall the representation of $f$ in \eqref{eq-D2-1}. It suffices to show that $g(x,y,p)=H_x(y,p)$ if $(y,p)\in A_x$ and $g(x,y,p)=\rho(x)$ if $(y,p)\notin A_x$.
Below we prove this. 
Fix   $x\le y$ and $p\in (0,1)$.
Let $B_{y,p}=\{t\in \R: (y,p)\in A_t\}$.
If $x,t\in B_{y,p}$,
then $H_t(y,p)=H_{x}(y,p)\ge \rho(x)$ by Proposition \ref{prop:H} part (iii). 
Hence, if $x\in B_{y,p}$, then
\begin{align}
\label{eq:proof4}
\sup_{t\in B_{y,p}} H_t(y,p)= H_{x}(y,p)\ge \rho(x).
\end{align}
If $x \not \in B_{y,p}$, then for any $t\in B_{y,p}$,  Proposition \ref{prop:H} part (iv) yields  $t<x$ and 
\begin{align}
\label{eq:proof5}
H_t(y,p)\le \rho(x).
\end{align}
Putting the above observations together, we obtain 
\begin{align*}
g(x,y,p)&=\max\left\{\sup_{t<x}\psi(t,0),\sup_{x\le t<y}\psi(t,p),\sup_{t>y}\psi(t,1)\right\}
\tag*{\footnotesize [definition of $g$]}
\\
&=\max\{\psi(x,0), \psi(y,p)\}
\tag*{\footnotesize [left continuity of $t\mapsto \psi(t,p)$ and \eqref{eq:proof3}]}
\\&=\max\{\rho(x),\psi(y,p)\}
\tag*{\footnotesize [using \eqref{eq:proof2}]}
\\
&=\rho(x)\vee \sup_{t\in\R}\left\{H_t(y,p)\id_{\{(y,p)\in A_t\}}-\infty\id_{\{(y,p)\notin A_t\}}\right\} \tag*{\footnotesize [definition of $\psi$]} \\
&=\rho(x)\vee \sup\{H_t(y,p): t\in B_{y,p}\} \tag*{\footnotesize [removing $-\infty$ from the supremum]} 
\\
&= H_x(y,p) \id_{\{x\in B_{y,p}\}} + \rho(x) \id_{\{x\not \in B_{y,p}\}}  . \tag*{\footnotesize [using \eqref{eq:proof4} and \eqref{eq:proof5}]} 
\\&= f(x,y,p). \tag*{\footnotesize [definition of $f$]}  
\end{align*} 
 This completes the proof.
  \end{proof}

\subsubsection{Proof on  $\mathcal M_D$}
In this section, we aim to show that the representation on $\mathcal M_{D,2}$ can be extended to $\mathcal M_D$.

For $F\in\mathcal M_D$, we assume that $F$ has the form $F=\sum_{i=1}^n p_i\delta_{x_i}$, where $p_i\in [0,1]$ for $i\in[n]$, $\sum_{i=1}^n p_i=1$ and $x_1\le \dots\le x_n$. Denote by $F_k=\sum_{i=1}^k p_i \delta_{x_1}+(1-\sum_{i=1}^k p_i)\delta_{k+1}$ for $k\in[n-1]$. One can easily check that $F=\bigvee_{k=1}^{n-2}F_k$. Applying Axiom \ref{ax:E} and the representation on $\mathcal M_{D,2}$ in Proposition \ref{prop:D2}, we have 
\begin{align*}
\rho(F)&=\max_{k\in[n-1]}\rho(F_k)
\\&=\max_{k\in[n-1]} \sup_{x\in\R} \psi(x,F_k(x))
\\&=\sup_{x\in\R}\max_{k\in[n-1]}  \psi(x,F_k(x))\\
&=\sup_{x\in\R} \psi\left (x,\min_{k\in[n-1]}F_k(x)\right)
=\sup_{x\in\R} \psi(x,F(x)),
\end{align*}
where the fourth step holds because $\psi$ is decreasing in its second argument (see Proposition \ref{prop:D2} (ii)). Therefore, the representation on $\mathcal M_{D}$ holds.

\subsubsection{Proof on $\mathcal M_C$}
Let $F\in \mathcal M_C$. Denote by $m=\inf\{x: F(x)>0\}$ and $M=\inf\{x: F(x)\ge 1\}$ the left and right support endpoint of $F$. Let $T=M-m$. Define a sequence of discrete distributions $\{F_n\}_{n\in\N}$ as 
\begin{align*}
	F_n(x)=\begin{cases}
		0,~~&x<m,\\
		F\left(m+\frac{k}{n}T\right),~~&m+\frac{k-1}{n}T\le x< m+\frac{k}{n}T,~k\in[n],\\
		1,~~&x\ge M.
	\end{cases}
\end{align*}
Using the representation on $\mathcal M_D$, we have,  for  all $n\in\N$,
\begin{align*}
\rho(F_n)=\sup_{x\in\R}\psi(x,F_n(x))
=\rho(m)\vee\max_{k\in[n]}  \psi\left(m+\frac{k}{n}T,F\left(m+\frac{k}{n}T\right)\right)\le \sup_{x\in\R}\psi(x,F(x)).
\end{align*}
Note that $F_n\lawto F$, and using Axiom \ref{ax:LS} yields 
\begin{align}\label{eq-con}
\rho(F)\le	\liminf_{n\to\infty}\rho(F_n)\le \sup_{x\in\R}\psi(x,F(x)).
\end{align}
On the other hand, we will verify that $\rho(F)\ge \psi(x,F(x))$ holds for any $x\in\R$. If $x< m$, then   $\psi(x,F(x))=\psi(x,0)=\rho(x)<\rho(m)\le \rho(F)$, where the last step follows from consistency with $\preceq_1$. If $x\ge M$, then  $\psi(x,F(x))=\psi(x,1)=-\infty$, which implies $\rho(F)\ge \psi(x,F(x))$.  If $x\in[m,M)$, define $G=F(x)\delta_m+(1-F(x))\delta_x$. 
It holds that $\rho(G)=\sup_{t\in\R}\psi(t,G(t))=\rho(m)\vee \psi(x,F(x))\ge \psi(x,F(x))$. One can easily check that $G\preceq_1 F$, and consistency with $\preceq_1$ implies $\rho(F)\ge \rho(G)\ge  \psi(x,F(x))$. Therefore, we conclude that $\rho(F)\ge \sup_{x\in\R}\psi(x,F(x))$. Combining with \eqref{eq-con}, we obtain $\rho(F)=\sup_{x\in\R}\psi(x,F(x))$. This completes the proof of (i) $\Rightarrow$ (ii) Theorem \ref{th-main}.

\subsection{Proof of Theorem \ref{th:main-MM}}
\label{sec:pf2}

\emph{\bf Sufficiency.} Note that $\rho(F)=\sup_{x\in\R}\psi(x,F(x))$ with $\psi(x,p)=f(x)-\infty\id_{\{p\ge \Lambda(x)\}}$. One can check that $\psi$ satisfies all conditions in Theorem \ref{th-main} (ii), which implies that $\rho$ satisfies Axioms \ref{ax:ND}, \ref{ax:LS} and \ref{ax:E}. To see Axiom \ref{ax:MS}, denote by $S_F=\{x: F(x)<\Lambda(x)\}$ for $F\in\mathcal M_C$. Let $F,G\in\mathcal M_C$, and denote by $H(x)=F\wedge G (x)=\max\{F(x),G(x)\}$. It holds that 
$S_H=S_F\cap S_G$. Since $\Lambda$ is decreasing, $S_F$ and $S_G$ are both intervals with zero as their left endpoint.
Hence, we have $S_F\subseteq S_G$ or $S_G\subseteq S_F$. If $S_F\subseteq S_G$, then $\rho(F)\le \rho(G)$ and
\begin{align*}
\rho(H)=\sup_{x\in S_H} f(x)=\sup_{x\in S_F\cap S_G} f(x)=\sup_{x\in S_F} f(x)=\rho(F)
=\rho(F)\wedge \rho(G).
\end{align*}
Similarly, we have $\rho(H)=\rho(G)
=\rho(F)\wedge \rho(G)$ when $S_G\subseteq S_F$. Hence, Axiom \ref{ax:MS} holds.

\emph{\bf Necessity.}
Suppose that $\rho$ satisfies Axioms \ref{ax:ND}, \ref{ax:LS}, \ref{ax:E} and \ref{ax:MS}.
Recall that Axiom \ref{ax:E} is stronger than consistency with $\preceq_1$, and 
together with  Axiom \ref{ax:ND}, $\rho$ satisfies Property \ref{ax:MC}.
We will frequently use these two properties in the proof.
By Theorem \ref{th-main}, we know that $\rho$ has the form:
\begin{align}\label{eq-mmeq}
	\rho(F)=\sup_{x\in\R}\psi(x,F(x)) \mbox{ for all $F\in\mathcal M_C$},
\end{align}
for some
function $\psi:\R\times[0,1]\to \R\cup\{-\infty\}$ that is increasing and lower semicontinuous in the first argument and
decreasing and lower semicontinuous in the second argument and satisfies  $\psi(x,0)<\psi(y,0)$ for $x<y$ and $\psi(x,1)=-\infty$ for all $x\in\R$. 
Define $\Lambda:\R\to[0,1]$ as follows: 
\begin{align*}
	\Lambda(x)=\sup\{p\in[0,1]: \psi(x,p)=\psi(x,0)=\rho(x)\}.
\end{align*}
Further, define the right-continuous inverse of $\Lambda$ as
\begin{align}\label{eq-Lambda-1}
	\Lambda^{-1+}(p)=\sup\{x: p<\Lambda(x)\}.
\end{align}
We use the convention that $\sup\emptyset:=-\infty$ and $\rho(-\infty):=\lim_{x\to-\infty}\rho(x)$. Since $x\mapsto\rho(x)$ is increasing, we also have $\rho(-\infty)=\inf_{x\in\R}\rho(x)$.
It suffices to verify the following results in order:
\begin{align*}
\begin{array}{ll}
\mbox{(i) $\psi(x,p)=\rho(x)$ if $p<\Lambda(x)$};~~&
\mbox{(ii) $\Lambda$ is decreasing;}\\
\mbox{(iii) $\psi(x,p)\le \rho(\Lambda^{-1+}(p))$ if  $p\ge \Lambda(x)$};~~& \mbox{(iv) $\Lambda>0$ on $\R$.}
\end{array}
\end{align*}
First, we show that \eqref{eq-mmeq} reduces to the wanted representation if all above statements hold. Let  $F\in\mathcal M_C$.
Note that Property \ref{ax:MC} and Axiom \ref{ax:LS} imply that $x\mapsto \rho(x)$ is strictly increasing and left-continuous. 
For any $t\in\{x:F(x)\ge \Lambda(x)\}$, it follows from the decreasing monotonicity of $\Lambda$ and the increasing monotonicity of $F$ that
\begin{align*}
\sup\{x:F(t)<\Lambda(x)\}\le \sup\{x:F(x)<\Lambda(x)\}\le t.
\end{align*}
It holds that
$$
\psi(t,F(t))\le \rho(\Lambda^{-1+}(F(t)))=\sup_{x:F(t)<\Lambda(x)}\rho(x) \le \sup_{x:F(x)<\Lambda(x)}\rho(x),
$$
where the first inequality follows from (iii) as $F(t)\ge \Lambda(t)$.
Therefore, 
\begin{align*}
\rho(F)&=\sup_{x\in\R}\psi(x,F(x))=\sup_{x:F(x)<\Lambda(x)}\rho(x)\vee\sup_{x:F(x)\ge\Lambda(x)}\psi(x,F(x))\\
&=\sup_{x:F(x)<\Lambda(x)}\rho(x)
=\sup\{\rho(x):F(x)<\Lambda(x)\},
\end{align*}
as wanted. Next, we verify the four statements  (i)--(iv) mentioned above.


(i) Noting that $\psi$ is decreasing in the second argument, 
Statement (i) is obvious by the definition of $\Lambda$.

(ii) For any $x,y\in\R$ with $x<y$ and $p\in[0,1]$,
it follows from the representation of $\rho$ in \eqref{eq-mmeq} that
\begin{align*}
\rho(x)=\psi(x,0)~~{\rm and}~~
\rho(p\delta_x+(1-p)\delta_y)=\rho(x)\vee \psi(y,p).
\end{align*}
Let $x,y,z\in\R$ with $x< y< z$ and define $F=p \delta_x+(1-p)\delta_{z}$. Using Axiom \ref{ax:MS}, we have
\begin{align}\label{eq-minSchain}
\rho(x)\vee \psi(y,p)
=\rho(p\delta_{x}+(1-p)\delta_{y})=\rho(F\wedge \delta_{y})=\rho(F)\wedge \rho(y)=\{\rho(x)\vee \psi(z,p)\}\wedge \rho(y).
\end{align}
Suppose that $p> \Lambda(y)$, which yields $\psi(y,p)<\rho(y)$. Also note that Property \ref{ax:MC} implies $\rho(x)<\rho(y)$. Hence, 
$\rho(x)\vee \psi(y,p)<\rho(y)$. Combining with \eqref{eq-minSchain} yields 
\begin{align}\label{eq-eqminS}
\rho(x)\vee \psi(y,p)=\rho(x)\vee \psi(z,p)~~{\rm if}~x<y<z~{\rm and}~p> \Lambda(y).
\end{align}
We now prove that $\Lambda$ is decreasing by contradiction. Suppose that there exist $y,z\in\R$ with $y<z$ such that $\Lambda(y)<\Lambda(z)$. Choosing $p\in(\Lambda(y), \Lambda(z))$, the definition of $\Lambda$ implies $\psi(y,p)<\rho(y)$ and $\psi(z,p)=\rho(z)$. Hence, for $x<y$ we have
\begin{align*}
	\rho(x)\vee\psi(z,p)=\rho(z)>\rho(x)\vee \rho(y)\ge \rho(x)\vee \psi(y,p).
\end{align*}
This contracts \eqref{eq-eqminS}. Thus, we conclude that $\Lambda$ is decreasing.

(iii) Note that $\psi$ is right-continuous in the second argument. It suffices to verify that $\psi(x,p)\le \rho(\Lambda^{-1+}(p))$ whenever $p>\Lambda(x)$. 
Define $T_p=\{t: p>\Lambda(t)\}$ for $p\in[0,1]$. For a fixed $p\in[0,1]$, we assert that there are only two cases to happen:
\begin{itemize}
\item Case 1: $\psi(y,p)=\psi(z,p)>\inf_{t\in\R}\rho(t)$ for all $y,z\in T_p$. In this case, $\psi(\cdot,p)$ is a constant on $T_p$, which is denoted by $h(p)$.

\item Case 2: $\psi(z,p)\le \inf_{t\in\R}\rho(t)$ for all $z\in T_p$.
\end{itemize}
Otherwise (note that $\psi$ is increasing in the first argument), there exist $y_0,z_0\in T_p$ with $y_0<z_0$ such that $\psi(y_0,p)<\psi(z_0,p)$ and $\psi(z_0,p)>\inf_{t\in\R}\rho(t)$. Let $x_0\in\R$ with $x_0<y_0$ and $\psi(z_0,p)>\rho(x_0)$. Hence, we have
\begin{align*}
\psi(z_0,p)>\rho(x_0)\vee\psi(y_0,p)=\rho(x_0)\vee\psi(z_0,p)=\psi(z_0,p),
\end{align*} 
where the first equality follows from  \eqref{eq-eqminS}. This yields a contradiction. It is not difficult to check that $T_p$ is nonempty if and only if $p>\inf_{t\in\R}\Lambda(t)$. Define  $B_1,B_2\subseteq(\inf_{t\in\R}\Lambda(t),1]$ as the sets such that Case 1 holds if $p\in B_1$  and Case 2 holds if $p\in B_2$. 
Note that $\psi$ is decreasing in the second argument. One can check that $B_1,B_2$ are two intervals that satisfies $\sup_{t\in B_1}t\le \inf_{t\in B_2}t$. Clearly, $\psi(x,p)\le \rho(\Lambda^{-1+}(p))$ for $p\in B_2$ and $x\in T_p$. If $B_1=\emptyset$, then (iii) holds directly. Suppose now $B_1\neq \emptyset$, and it remains to verify that $h(p)=\psi(x,p)\le \rho(\Lambda^{-1+}(p))$ for $p\in B_1$ and $x\in T_p$.
By the definitions of $T_p$ and $\Lambda$, we have $h(p)=\psi(x,p)<\rho(x)$ for all $p\in B_1$ and $x\in T_p$.
Define $\Lambda^{-1}(p)=\inf\{x: \Lambda(x)<p\}$ as the left-continuous inverse of $\Lambda$. For any $p\in B_1$, it holds that $p>\Lambda(\Lambda^{-1+}(p)+\epsilon)$ for all $\epsilon>0$, and hence,
\begin{align*}
h(p)=\psi(\Lambda^{-1+}(p)+\epsilon,p)<\rho(\Lambda^{-1+}(p)+\epsilon) \mbox{ for all $\epsilon>0$}.
\end{align*}
This further implies
\begin{align}\label{eq-hle}
	h(p)\le \rho_+(\Lambda^{-1}(p)),~~{\rm where}~ \rho_+(x)=\lim_{t\downarrow x}\rho(t).
\end{align}
On the other hand, since $\psi$ is increasing in the first argument, we have for some $x_0\in T_p$
\begin{align*}
h(p)=\psi(x_0,p)\ge \psi(t,p)=\rho(t) \mbox{ for all $t\in\R$ with $p<\Lambda(t)$},
\end{align*}
where the last step follows from Statement (i). Hence, we have
$
h(p)\ge \rho(\Lambda^{-1+}(p)-\epsilon)
$
for all $\epsilon>0$,
where $\Lambda^{-1+}$ defined in \eqref{eq-Lambda-1} is the right-continuous inverse of $\Lambda$.
Note that $x\mapsto\rho(x)$ is increasing and left-continuous, and we have
$	h(p)\ge \rho(\Lambda^{-1+}(p))$ for all $p\in B_1$.
Combining with \eqref{eq-hle} yields
\begin{align*}
\rho(\Lambda^{-1+}(p))	\le	h(p)\le \rho_+(\Lambda^{-1}(p)) \mbox{ for all $p\in B_1$}.
\end{align*}
Since $x\mapsto\rho(x)$ is increasing and left-continuous, $\rho_+(x)$ is increasing and right-continuous, $\Lambda^{-1}$ is decreasing and left-continuous and $\Lambda^{-1+}$ is decreasing and right-continuous, one can check that $\rho(\Lambda^{-1+}(p))$ is the right-continuous version of $\rho_+(\Lambda^{-1}(p))$ on $B_1$. Also note that the lower semicontinuity of $\psi$ in its second argument implies that $h$ is right-continuous on $B_1$. Hence, we have
\begin{align*}
	h(p)=\rho(\Lambda^{-1+}(p)) \mbox{ for all $p\in B_1$}.
\end{align*}
This competes the proof of Statement (iii).

(iv) We assume by contradiction that $\{x: \Lambda(x)=0\}\neq\emptyset$ and define
$c:=\inf\{x: \Lambda(x)=0\}$. Since $\Lambda$ is decreasing, we have that $\Lambda(x)>0$ for $x<c$ and $\Lambda(x)=0$ for $x>c$.
By the previous results, we know that 
\begin{align*}
\rho(F)=\sup\{\rho(x): F(x)<\Lambda(x)\} \mbox{ for all $F\in\mathcal M_C$}.
\end{align*}
Following standard calculation, one can get $\rho(c+1)=\rho(c)$, which contradicts Property \ref{ax:MC}. Hence, we complete the proof.



\end{document}